\documentclass[11pt,a4paper]{article}

\usepackage[utf8]{inputenc}
\usepackage[T1]{fontenc}
\usepackage{mathptmx}

\usepackage{amsmath,amssymb,amsthm,bm,dsfont}

\usepackage{graphicx}
\usepackage{subcaption}
\usepackage{multirow}
\usepackage{nicematrix}
\usepackage{dcolumn}

\usepackage[a4paper,margin=1in]{geometry}

\usepackage{enumitem}
\usepackage{pifont}

\usepackage{tikz}

\usepackage[round,authoryear]{natbib}
\usepackage[colorlinks=true,allcolors=black,citecolor=black]{hyperref}

\usepackage{etoolbox}
\usepackage{comment}
\usepackage{fix-cm}

\newlist{todolist}{itemize}{2}
\setlist[todolist]{label=$\square$}

\newcommand{\E}{\operatorname{E}}
\newcommand{\Var}{\operatorname{Var}}
\newcommand{\Cov}{\operatorname{Cov}}

\newtheorem{theorem}{Theorem}
\newtheorem{lemma}{Lemma}
\newtheorem{proposition}{Proposition}
\newtheorem{corollary}{Corollary}
\newtheorem{definition}{Definition}
\newtheorem{remark}{Remark}
\newtheorem{assumption}{Assumption}

\newcommand{\ignore}[1]{}

\date{}

\begin{document}

\begin{center}
    \vspace*{-2em}
    {\LARGE\bfseries Ordinal Patterns Based Testing of Spatial Independence in Irregular Spatial Structures\par}
    \vspace{1.5em}

    {\large
    \textbf{Giorgio Micali}$^{1}$,
    \textbf{David Garnés-Galindo}$^{2}$,
    \textbf{Mariano Matilla-García}$^{3}$,
    \textbf{Manuel Ruiz-Marín}$^{2}$\par}
    \vspace{1em}

    {\small
    $^{1}$Department of Applied Mathematics, University of Twente, Enschede, The Netherlands\\[0.5em]
    $^{2}$Departamento de Métodos Cuantitativos, Universidad Politécnica de Cartagena, Cartagena, Spain\\[0.5em]
    $^{3}$Facultad de Económicas y Empresariales, UNED, Madrid, Spain\\[0.8em]
    \texttt{g.micali@utwente.nl, david.garnes@upct.es, mmatilla@cee.uned.es, manuel.ruiz@upct.es}\par}
    \vspace{2em}
\end{center}
\begingroup
\renewcommand\thefootnote{}\footnotetext{Copyright (2026) Author(s). This article is distributed under a Creative Commons Attribution-NonCommercial-NoDerivs 4.0 International License (CC BY-NC-ND 4.0).}
\addtocounter{footnote}{-1}
\endgroup
\begin{abstract}
We propose a nonparametric test of spatial independence for data observed
on irregular, non-lattice point clouds $\mathcal{V}_{n}\subset\mathbb{R}^{2}$.
For each location $v\in\mathcal{V}_{n}$, we encode the local spatial
configuration through the ordinal pattern of the $m$ nearest-neighbour
observations, obtaining a symbolic representation that is invariant under
strictly monotone transformations and robust to outliers.
Under the null hypothesis of spatial independence, the local ordinal patterns
are i.i.d.\ and uniformly distributed over the symmetric group $\mathcal{S}_{m}$,
regardless of the unknown marginal distribution $F$.
We exploit this characterisation to construct a test statistic $L_{n}$
based on the additive log-ratio (ALR) transformation of the empirical
ordinal-pattern frequencies.
Invoking a central limit theorem for graph-dependent processes under a
graph-based $\alpha$-mixing condition, we establish that $L_{n}$ converges
in distribution to a $\chi^{2}_{m!-1}$ random variable, yielding an
asymptotically pivotal procedure with no nuisance parameters.
An extensive Monte Carlo study confirms that the $\chi^{2}_{m!-1}$
approximation is accurate already at moderate sample sizes, that the test
controls size at the nominal level, and that power increases monotonically
with the strength of spatial dependence.
Notably, the test detects dependence in both linear and nonlinearly transformed
spatial autoregressive models, illustrating the robustness that is
characteristic of ordinal-pattern methods.
Our framework extends the spatial ordinal-pattern testing paradigm from
regular lattices to general spatial supports, opening the door to
ordinal-pattern inference in the many applied settings where observations
are irregularly located.
\end{abstract}

\section{Introduction}
Spatial dependence plays a central role in scientific analysis because spatially
indexed observations are rarely independent. Proximity and connectivity induce
interactions that violate classical statistical assumptions, leading to invalid
inference if ignored. Beyond its methodological implications, spatial dependence
reflects fundamental generative mechanisms such as diffusion, spillovers, and
contagion. Testing for spatial dependence is therefore essential both to ensure
statistical validity and to uncover the processes governing spatial systems.
 
To formalize these ideas, consider a spatial domain represented by a finite set
of locations $\mathcal V_n = \{ v_1, \ldots, v_n\} \subset \mathbb{R}^2$, where each point
denotes the spatial coordinates of an observational unit, represented as
\[
  v_i = [v_{i1}, v_{i2}]^\top.
\]
At each location $v_s \in \mathcal \mathcal{V}_n$ we observe a real-valued measurement $X_s$. For
instance, $\mathcal{V}_n$ may correspond to the positions of shops within a
city and $X_s$ to the price of a specific product at location $v_s$. Spatial
dependence in the field $(X_s)_{v_s \in \mathcal V_n}$ may arise from competition among
shops, demographic variation, accessibility, or other spatially structured
interaction mechanisms.
 
Our goal is to develop a nonparametric statistical test for \emph{spatial
independence}. Formally, we consider
\[
\begin{aligned}
\mathcal{H}_0:\;& (X_s)_{v_s \in \mathcal V_n}\ \text{i.i.d.} \\
&\text{vs.} \\
\mathcal{H}_1:\;& (X_s)_{v_s \in \mathcal V_n}\ \text{spatially dependent.}
\end{aligned}
\]
 
The null hypothesis of spatial independence constitutes a canonical reference
model across a wide range of scientific disciplines, representing the absence of
interaction or coupling among spatial units. In the social sciences, rejecting
this null indicates the presence of spillovers, competition, or peer effects
\citep{clifford1981}; in ecology and epidemiology, it corresponds to dispersal
or contagion mechanisms \citep{Ripley1981,Cressie1993}; and in statistical
physics and complex systems, it reflects the emergence of correlations in
interacting random fields \citep{KindermannSnell1980,Newman2010}. Rejecting
$\mathcal{H}_0$ therefore corresponds to detecting structured spatial
organization---such as clusters, gradients, or anisotropy---arising from
interaction-driven processes.
 
We adopt a symbolic approach: local spatial configurations are encoded through
\emph{ordinal patterns}, and spatial dependence is inferred from deviations of
the empirical ordinal-pattern distribution from the uniform distribution
predicted under $\mathcal{H}_0$. Ordinal methods are nonparametric, invariant
under strictly monotone transformations, and robust to outliers and
distributional misspecification.
 
Ordinal-pattern (OP) methods are well-established in time-series analysis as
tools for detecting serial dependence. In recent years, there has been growing
interest in extending OP-based methodology to spatial settings. However,
existing works restrict attention to the case where data are observed on a
\emph{regular lattice}: the spatial field $(X_{\mathbf{t}})_{\mathbf{t} \in
\mathbb{Z}^2}$ is sampled on a rectangular grid, and spatial ordinal patterns
(SOPs) are constructed by ranking observations within $2 \times 2$ blocks of
neighboring sites. This framework underlies the SOP tests developed by Wei\ss{}
and coauthors, including entropy-based SOP tests~\citep{Weiss20241},
nonparametric SOP tests for spatial
dependence~\citep{weiß2025nonparametrictestingspatialdependence}, and OP-based
tests via Hilbert space-filling curves for higher-dimensional grid
data~\citep{weiss_2025}. These methods exploit the fact that, under
$\mathcal{H}_0$, SOPs from $2 \times 2$ lattice windows are i.i.d.\ uniformly
distributed over the $4!$ possible permutations, and deviations from uniformity
signal spatial dependence.
 
However, most existing SOP-based approaches fundamentally rely on the
\emph{reticular} structure of the underlying sampling scheme: observations must
lie on a regular grid in $\mathbb{R}^2$ (or $\mathbb{R}^3$), and local pattern
formation depends on fixed $2 \times 2$ windows or on the traversal of grid
points by a space-filling curve. These constructions implicitly impose artificial
symmetries and uniform neighborhood relations that rarely hold in empirical
systems.
 
Many spatial systems of scientific interest are instead defined on irregular
point clouds, where observational units are unevenly spaced, heterogeneously
distributed, or connected through nonuniform interaction structures. Spatial
interactions in such systems are governed by heterogeneous geometry and topology
rather than by regular Euclidean grids, and local connectivity patterns may vary
substantially across space. Concrete examples include shop locations in urban
environments, meteorological or environmental monitoring stations, and sensor
networks whose spatial placement is determined by logistical, geographical, or
infrastructural constraints. In all these settings, neighborhood size,
connectivity, and spatial scale vary across locations, and the notion of a fixed
$2 \times 2$ block is inherently absent. As a consequence, SOP methodologies
designed for regular lattices are not applicable, motivating the development of
inference tools capable of detecting spatial dependence on irregular spatial
supports \citep{Ripley1981,Cressie1993,Baddeley2015}.
 
Motivated by this need, we propose what appears to be the first
ordinal-pattern-based test of spatial independence for data sampled on
\emph{irregular, non-lattice point clouds}. Instead of relying on rectangular
lattice windows, we construct local ordinal patterns from nearest-neighbor
configurations for each $v \in \mathcal{V}_n \subset \mathbb{R}^d$. Under $\mathcal{H}_0$,
the neighborhood vectors consist of i.i.d.\ continuous observations, so their
ordinal patterns are i.i.d.\ uniform over $\mathcal{S}_m$. Departures from this
uniform distribution serve as diagnostic evidence for spatial dependence.
 
Our work thus extends the SOP testing paradigm from regular lattices to general
point clouds in $\mathbb{R}^d$, while preserving the key advantages of ordinal
methods: nonparametricity, invariance under monotone transformations, and
robustness to outliers and distributional contamination.

The remainder of the paper is organized as follows. In Section \ref{sec:mathfra}, we introduce the mathematical framework. Section \ref{sec:theorem} presents the spatial mixing condition and establishes a central limit theorem for ordinal patterns. In Section \ref{sec:cov}, we develop a consistent estimator of the covariance matrix. Section \ref{sec:test_statistics} introduces a test statistic for spatial independence along with its asymptotic distribution. In Section \ref{sec:monte}, we examine the finite-sample performance of the test statistic through an extensive Monte Carlo simulation study. Finally, Section \ref{sec:conclusions} concludes the paper.

\section{Mathematical Framework}\label{sec:mathfra}

For each $n\in\mathbb N$, let $\mathcal{M}_n:=\{1,\dots,n\}$ index a collection of spatial locations
$$ \mathcal{V}_n = \left\{v_{s}=\left[\begin{array}{c}
     v_{1s}  \\
     v_{2s}
\end{array} \right]\in\mathbb R^2:s\in \mathcal{M}_n \right\}$$ and associated observations
$\{X_{s}:v_s\in\mathcal V_n\}$.
Fix an integer $m\ge 2$, which we will referred to as \emph{embedding dimension}. For each node $v_s\in\mathcal V_n$,
let $B(s)\subset\mathcal V_n$ denote the set consisting of $v_s$ and its $(m-1)$ closest locations
according to Euclidean distance, i.e.
\[
B(s)
:=
\{v_s\}\,\cup\,
\operatorname*{arg\,min}_{\substack{
A \subset \mathcal V_n \setminus \{v_s\} \\
|A| = m-1
}}
\sum_{v_u \in A} \bigl\| v_{u} - 
v_{s} \bigr\|,
\]
with polar-angle ties broken deterministically. Then $$B(s)=\{v_s,v_{s_1},\dots,v_{s_{m-1}}\}$$
where the elements in $B(s)$ are ordered by increasing Euclidean distance to $v_s$, or by positive polar angles in case of ties. 

We then define the undirected $m$-nearest-neighbor graph
$G_n=(\mathcal V_n, \mathcal{E}_n)$ by declaring an edge between $v_s$ and $v_u$ whenever $v_u$ belongs to the
local neighborhood of $v_s$ (or vice versa), that is,
\[
\{v_s,v_u\}\in \mathcal{E}_n
\quad\Longleftrightarrow\quad
v_u\in B(s)\setminus\{v_s\}\ \ \text{or}\ \ v_s\in B(u)\setminus\{v_u\}.
\]
Let $d_n(\cdot,\cdot)$ denote the corresponding graph distance on $G_n$, i.e.\ the length of the shortest path in $G_n$.

To each location $v_s\in \mathcal{V}_n$ we define the local block
$X_{B(s)}:=(X_{u})_{v_u\in B(s)}$. Let $\mathcal S_m$ denote the symmetric group of order $m!$, that is, the group formed by all the permutations of the set $\{1, \ldots, m\}$. 

To each location $v_{s_0}\in \mathcal{V}_n$ we associate a permutation $\pi=(r_0,r_1,\dots,r_{m-1})\in \mathcal{S}_m$ satisfying
$$X_{s_{r_0}}\leq \,X_{s_{r_1}}\leq\,\dots\leq \,X_{s_{r_{m-1}}}\;.$$ and denote it as 
\[
\Pi(X_{B(s_0)}):=(r_0,r_1,\dots,r_{m-1})\in \mathcal S_m,
\]
where $\Pi(\cdot)$ returns the permutation describing the relative order of its arguments.
For a fixed $\pi\in \mathcal S_m$, we finally define the associated ordinal-pattern indicator variable
\begin{equation}
\label{eq:variables_Y}
    Y_{s}:=\mathbf 1\{\Pi(X_{B(s)})=\pi\},
\qquad v_s\in\mathcal V_n.
\end{equation}
Therefore, our object of interest is the spatial process $(X_{s})_{v_s\in \mathcal V_n}$ and its associated graph-process $G_n=\left( \mathcal V_n, \mathcal{E}_n, (Y_{s})_{v_s\in \mathcal V_n}\right)$. An illustration  can be found in Figure \ref{fig:OP_graph}. 
\begin{figure*}
    \centering
    \includegraphics[width=0.8\linewidth]{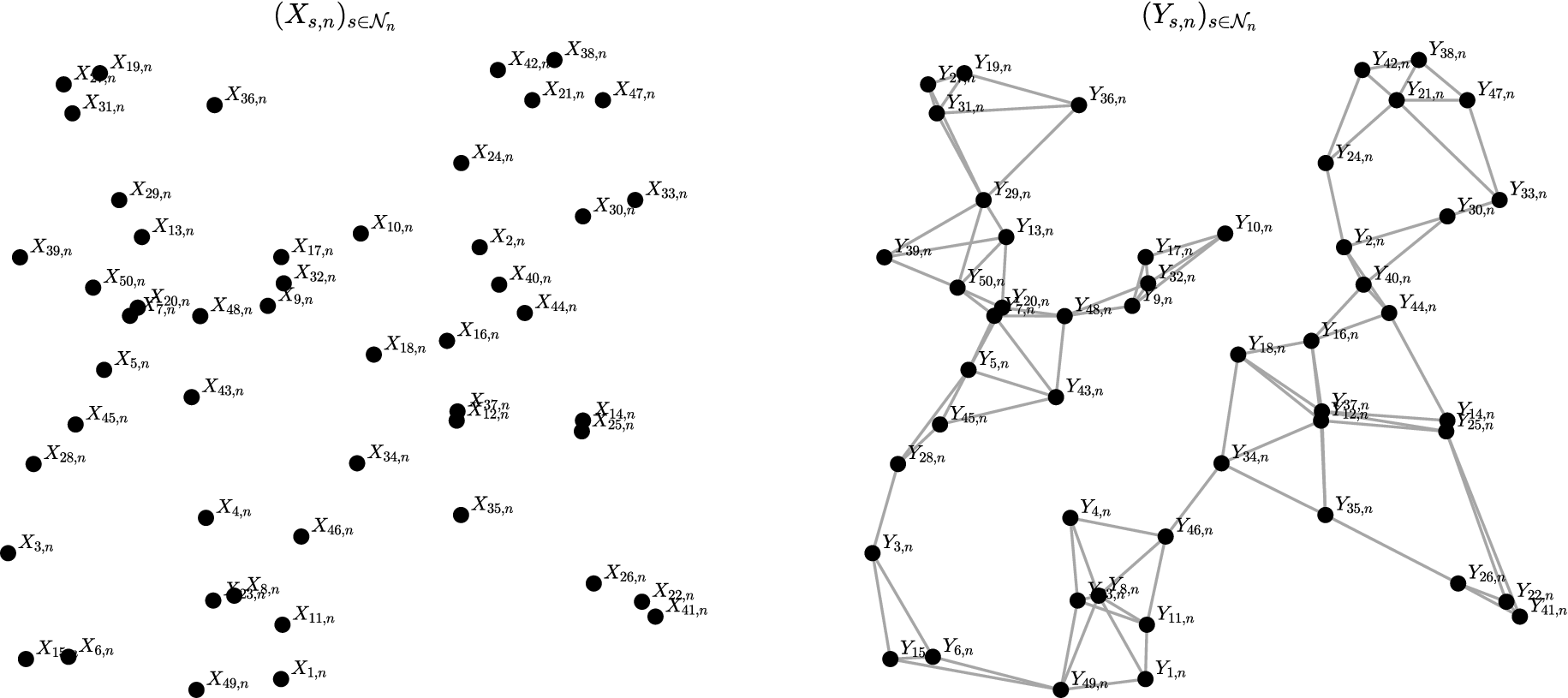}
    \caption{Ordinal patterns graph construction from initial process. }
    \label{fig:OP_graph}
\end{figure*}

\section{central limit theorem for spatial ordinal patterns}\label{sec:theorem}
In this section we investigate under which conditions a central limit theorem (CLT) for $\sum_{v_s\in \mathcal{V}_n} Y_{s}$ holds. Let us denote $\mathcal S_m=\{\pi_1, \ldots, \pi_{m!}\}$ and define the vector of ordinal patterns indicators as 
\begin{equation}
    \label{eq:Y_vector}
\mathbf Y_{s}
:=
\begin{pmatrix}
\mathbf 1\{\Pi(X_{B(s)})=\pi_1\}\\
\vdots\\
\mathbf 1\{\Pi(X_{B(s)})=\pi_{m!-1}\}
\end{pmatrix}\in\{0,1\}^{m!-1}.
\end{equation}

We work with $m!-1$ (rather than $m!$) coordinates because for each fixed $s$ the indicators satisfy
\[
\sum_{j=1}^{m!}\mathbf 1\{\Pi(X_{B(s)})=\pi_j\}=1\qquad\text{a.s.},
\]
so the $m!$-dimensional indicator vector has a degenerate covariance matrix. Dropping one coordinate
removes this linear constraint, while still determining the ordinal-pattern distribution.

To establish asymptotic results for sums of dependent random vectors, one must quantify and control the dependence they exhibit. In the time-series setting, a standard approach is to impose suitable moment conditions together with a \emph{mixing} assumption, for example \(\alpha\)-mixing, meaning that, for \((X_t)_{t\in\mathbb Z}\) a time series and \(\sigma(\cdot)\) denoting the \(\sigma\)-field generated by the indicated random variables,
\begin{align*}
    \alpha(h)=\sup_{k\in\mathbb Z}\alpha\bigl(\sigma(X_t:t\le k),\sigma(X_t:t\ge k+h)\bigr)\xrightarrow{h \to \infty} 0\;.
\end{align*}
Thus, the dependence between the past and the future weakens as their temporal separation, measured by the lag \(h\), increases. We refer the reader to \cite{Bradley2005} for a thorough overview of mixing coefficients and their technical foundations.

In spatial settings, however, there is no canonical linear ordering of the observations. As a result, dependence can no longer be described solely in terms of temporal lags. Instead, it must also reflect the arrangement of the vertices in space, and is therefore encoded through a graph that specifies which variables may interact and how separation is measured.  A natural analogue of large time lag is then large graph
distance: variables associated with vertices that are far apart in the graph should become nearly
independent. Mixing coefficients that depend on subsets of locations on dates back to \cite{Bulinskii1988}, and were further expanded by \cite{Bolthausen} and \cite{JENISH200986}. In this work, we adopt the notion of strong-mixing processes proposed by \cite{KOJEVNIKOV2021882}, which is reported as Definition 
\ref{def:cond_alpha_mixing} below.

\begin{definition}\label{def:cond_alpha_mixing}
Let \(\{Z_{v}\}_{v\in\mathcal V_n}\) be random variables 
and let \( \mathcal{F}\) be a \(\sigma\)-field.
For two subsets \(A,B\subset \mathcal V_n\), define the \(\sigma\)-fields
\[
\mathcal F_{A}:=\sigma\!\big(Z_{v}:v\in A\big),\qquad
\mathcal F_{B}:=\sigma\!\big(Z_{v}:v\in B\big).
\]
Let \(\mathrm{dist}(A,B)\) denote the graph distance between \(A\) and \(B\), i.e.
\[
\mathrm{dist}(A,B):=\min\{d_n(v_i,v_j):v_i\in A,\ v_j\in B\}.
\]
The  \(\alpha\)-mixing coefficient at distance \(h\ge 0\) is
\[
\alpha_n(h)
:=
\sup_{\substack{A,B\subset\mathcal V_n\\ \mathrm{dist}(A,B)\ge h\\E\in\mathcal F_{A},\,F\in\mathcal F_{B}}}
\Bigl|
\mathbb P(E\cap F )-\mathbb P(E )\,\mathbb P(F  )
\Bigr|.
\]
We say that \(\{Z_{v}\}_{v\in\mathcal V_n}\) is \(\alpha\)-mixing 
if \(\alpha_n(h)\to 0\) as \(h\to\infty\).
\end{definition}
Unlike the case of alpha mixing for time series, we have an extra dependence, i.e., dependence on both $h$ and $n$.  Many processes on graph are $\alpha$ mixing. Example \ref{ex:SAR_alpha_mixing} shows that this is the case for the widely used class of SAR models on graph (see \cite{LeSagePace2009}).  

\noindent
\textbf{Example.}
\label{ex:SAR_alpha_mixing}
Let \(G_n=(\mathcal V_n,\mathcal{E}_n)\) be an undirected graph, and let \(W=(w_{ij})_{v_i,v_j\in\mathcal V_n}\) be a weight matrix such that
\[
w_{ij}=0 \qquad \text{whenever } v_i\neq v_j \text{ and } (v_i,v_j)\notin \mathcal{E}_n.
\]
Assume moreover that \(W\) is row-stochastic, i.e.
\[
\sum_{v_i\in\mathcal V_n} |w_{ij}| = 1
\qquad \text{for all } j\in\mathcal M_n,
\]
and let \(\rho\in(-1,1)\). Consider the Gaussian SAR model $ \mathbf X=(X_{s})_{v_s\in\mathcal V_n}$ defined by
\begin{equation}
    \mathbf X=\rho W \mathbf X+\varepsilon
\end{equation}
and in reduced form
\begin{equation}
\label{eq:SAR}
\mathbf X=(I_n-\rho W)^{-1}\mathbf \varepsilon,
\end{equation}
where \(\mathbf \varepsilon=(\varepsilon_{s})_{v_s\in\mathcal V_n}\) has i.i.d.\ centered Gaussian entries with finite variance \(\sigma^2<\infty\). Then \(\mathbf X\) admits the expansion
\begin{equation}
\label{eq:SAR_series}
   \mathbf X=\sum_{k=0}^\infty \rho^k W^k \mathbf \varepsilon, 
\end{equation}
where the series converges absolutely in \(L^2\). Moreover, if \(v_i,v_j\in\mathcal V_n\) satisfy \(d_n(v_i,v_j)=h\), then
\[
\bigl(W^k\bigr)_{ij}=0
\qquad\text{for all } k<h.
\]
Using \eqref{eq:SAR_series}, the fact that $\|W\|_\infty\leq 1$ and that $\mathbf{\varepsilon}$ is centered zero-mean Gaussian vector with finite variance, there exists a constant $C>0$ such that
\[
|\Cov(X_i,X_j)|
\le C \sum_{m\ge h} |\rho|^m
\le C \frac{|\rho|^h}{1-|\rho|}.
\]
Hence there exist constants \(C'>0\) and \(\eta\in(0,1)\), independent of \(n\), such that
\[
|\Cov(X_i,X_j)|\le C'\,\eta^{d_n(v_i,v_j)}
\qquad\text{for all }v_i,v_j\in\mathcal V_n.
\]
Thus the covariance decays exponentially fast with the graph distance.

On the other hand the minimum eigenvalue of $\Sigma$ we proceed as follows. Let $\Sigma =\Cov(\mathbf{X})$. From \eqref{eq:SAR}
\[
\Sigma
=
\sigma^2 (I_n-\rho W)^{-1}(I_n-\rho W)^{-T},
\]
thus we have
\[
\Sigma^{-1}
=
\frac{1}{\sigma^2}(I_n-\rho W)^T(I_n-\rho W).
\]
it follows that 
\[
\lambda_{\min}(\Sigma)
=
\frac{1}{\lambda_{\max}(\Sigma^{-1})}
=
\frac{\sigma^2}{\lambda_{\max}\!\bigl((I_n-\rho W)^T(I_n-\rho W)\bigr)}.
\]
Here and throughout, \(\|\cdot\|_{2}\) denotes the operator norm induced by the Euclidean norm, i.e,
$\|A\|_2:=\sup\limits_{\|x\|_2=1}\|Ax\|_2\;.$
By the standard identity
\[
\lambda_{\max}(A^\top A)=\|A\|_{2}^2\;,
\]
and using the fact that \(W\) is symmetric, it follows that
\[
(I_n-\rho W)^T(I_n-\rho W)=(I_n-\rho W)^2,
\]
and hence
\[
\lambda_{\max}\!\bigl((I_n-\rho W)^T(I_n-\rho W)\bigr)
=
\|I_n-\rho W\|_{2}^2.
\]
Using the triangle inequality and \(\|W\|_{2}\le 1\),
\[
\|I_n-\rho W\|_{2}
\le \|I_n\|_{2}+|\rho|\,\|W\|_{2}
\le 1+\rho.
\]
Thus
\[
\lambda_{\max}\!\bigl((I_n-\rho W)^T(I_n-\rho W)\bigr)
\le (1+\rho)^2,
\]
and consequently
\[
\lambda_{\min}(\Sigma)
\ge \frac{\sigma^2}{(1+\rho)^2}>0.
\]
This bound is uniform in \(n\).
Since the $\mathbf X$ is Gaussian and nondegenerate, it follows that 
Proposition~\ref{prop:gaussian_alpha_bound} in Appendix \ref{appendix:auxiliary} applies with
\[
M=C'\,\eta^h,
\]
we obtain
\begin{align*}
    \alpha_{n}(h)
:=&
\sup_{\substack{A,B\subset\mathcal V_n\\ |A|\le a,\ |B|\le b\\  \mathrm{dist}(A,B)\ge h}}
\alpha\!\bigl(\sigma(X_{u}:u\in A),\sigma(X_{v}:v\in B)\bigr)
\\&\le O\left( n\,\eta^h \right).
\end{align*}
Therefore the Gaussian SAR model is size-restricted \(\alpha\)-mixing on the graph, with exponentially decaying mixing coefficients.

In the following, we introduce the mathematical objects needed to formalize the denseness properties of the graph, following the framework of \cite{KOJEVNIKOV2021882}. Their results show that a central limit theorem holds under two complementary conditions: (i) first, a sparsity condition preventing the graph from becoming too dense (for example, by requiring bounded or only slowly growing neighborhood sizes), and (ii) a decay condition ensuring that dependence weakens sufficiently fast as the graph distance increases. In particular, these two requirements must be balanced: if the graph becomes denser, then the decay of dependence with distance must be correspondingly stronger in order for a central limit theorem to remain valid.

For \(v_i\in\mathcal V_n\) and \(h\in\mathbb N\cup\{0\}\), let
\[
\mathcal V_n(v_i;h):=\{v_j\in \mathcal{V}_n:d_n(v_i,v_j)\le h\}
\]
denote the \(h\)-neighborhood of \(v_i\).
The corresponding \(h\)-th \emph{neighborhood shell} is
\begin{equation}\label{eq:def_shell}
\mathcal V^{\partial}_n(v_i;h)
:=
\mathcal V_n(v_i;h)\setminus \mathcal V_n(v_i;h-1),
\qquad h\ge 0,
\end{equation}
with the convention \(\mathcal V_n(v_i;-1):=\varnothing\). Equivalently,
\[
\mathcal V^{\partial}_n(v_i;h)=\{v_j\in \mathcal{V}_n:d_n(v_i,v_j)=h\}.
\]
Next, we introduce a measure that quantifies, on average, how quickly the number of vertices
at graph distance \(h\) from a typical node grows; it is thus a simple summary of the sparsity
of the network at radius \(h\). Such quantity is denoted by \(\delta_n^{\partial}(h;k)\) and defined as
\begin{equation}\label{eq:KMS_3_1}
\delta^{\partial}_n(h;k)
:=
\frac{1}{n}\sum_{v_i\in\mathcal V_n}\bigl| \mathcal V^{\partial}_n(v_i;h)\bigr|^{\,k}\;.
\end{equation}
 for  $k,h \in \mathbb{N}$.
Further, we want to capture the potential overlap and local clustering of neighborhoods, that is, how strongly
dependence can propagate through nearby network connections. To this end, we introduce a quantity that controls
how large the \(f\)-neighborhood of a node can remain after removing the part that is already “close”
(within distance \(h-1\)) to some vertex in its \(h\)-shell. This quantity is defined as
\begin{equation}\label{eq:KMS_3_2}
\Delta_n(h,f;k)
:=
\frac{1}{n}\sum_{v_i\in\mathcal V_n}
\max_{v_j\in \mathcal V^{\partial}_n(v_i;h)}
\bigl| \mathcal V_n(v_i;f)\setminus \mathcal V_n(v_j;h-1)\bigr|^{\,k}.
\end{equation}
Finally, we combine \eqref{eq:KMS_3_1} and \eqref{eq:KMS_3_2} into a single quantity that captures both
sparsity  and neighborhood overlap simultaneously.
\begin{equation}\label{eq:KMS_3_3}
c_n(h,f;k)
:=
\inf_{\beta>1}
\Bigl[\Delta_n\!\bigl(h,f;k\beta\bigr)\Bigr]^{1/\beta}
\Bigl[\delta^{\partial}_n\!\Bigl(h;\frac{\beta}{\beta-1}\Bigr)\Bigr]^{1-1/\beta}.
\end{equation}
The definition of \(c_n(h,f;k)\) uses a Hölder-type interpolation: the exponent \(\beta>1\) allows us to trade off
control of \(\Delta_n\) against control of \(\delta_n^{\partial}\) by splitting a product into two moments with
conjugate powers \(\frac{1}{\beta}\) and \(\frac{\beta-1}{\beta}\).
Taking the infimum over \(\beta>1\) selects the most favorable such trade-off for the given network, yielding the
optimal bound.

Before stating the main theorem of the paper we requiere the following assumption, adapted from condition ND in \cite{KOJEVNIKOV2021882}.
\begin{assumption}
\label{assumption1}
    There exist \(p>4\) and a sequence \(f_n\to\infty\) such that, the mixing coefficients $(\alpha_n (h))_{h,n}$ of $(X_{s})_{v_s \in \mathcal V_n}$ and the graph $G$ satisfy 
\begin{enumerate}[label=(\alph*)]
\item \(\alpha_{n}^{\,1-1/p}(f_n)=o_{\mathrm{a.s.}}(n^{-3/2})\); 
\item for each \(k\in\{1,2\}\),
\[
\frac{1}{n^{k/2}}\sum_{h\ge 0} c_n(h,f_n;k)\,\alpha_{n}^{\,1-\frac{k+2}{p}}(h)
= o_{\mathrm{a.s.}}(1);
\]
\end{enumerate}

\end{assumption}

\begin{theorem}\label{thm:cond_vector_clt}
Suppose that the spatial process \((X_s)_{v_s\in\mathcal V_n}\) satisfies Assumption~\ref{assumption1}. Fix \(m\ge 2\), and let \((\mathbf Y_s)_{v_s\in\mathcal V_n}\) denote the \((m!-1)\)-dimensional random vectors defined in \eqref{eq:Y_vector}. For \(r\ge 0\), define $
U_n(h):=\#\{(v_s,v_t)\in\mathcal V_n \times \mathcal{V}_n:\ d_n(v_s,v_t)=h\},$
that is, the number of vertices in \(\mathcal V_n\) at graph distance \(h\) from \(v_s\). Assume that there exist constants \(C_2>C_1>0\), independent of \(n\), such that for all \(n\), 
\begin{align}\label{eq:eig}
&\min_{v_s\in\mathcal V_n}\lambda_{\min}\!\Big(\Var(\mathbf Y_{s}  )\Big)\ \ge\ C_1
\qquad\text{a.s.}\tag{Eig}\\
&\sum_{r\ge 1}\alpha_n(h)\,U_n(h)\le C_2n\;.
\label{eq:A_mixing_small}\tag{Mix}
\end{align}
Define the centered partial sum and its variance as
\[
S_n
:=
\sum_{v_s\in\mathcal V_n}
\Big(\mathbf Y_{s}-\E[\mathbf Y_{s}  ]\Big)
\quad \mbox{ and } \quad 
\Sigma_n:=\Var(S_n  )
\]
respectively.
Then \(\Sigma_n\) is almost surely positive definite, and the following multivariate CLT holds:
\begin{equation}\label{eq:cond_vector_clt}
\sup_{\mathbf x\in\mathbb R^{m!-1}}
\left|
\mathbb P\!\left(\Sigma_n^{-1/2}S_n \le \mathbf x   \right)
-
\Phi(\mathbf x)
\right|
\xrightarrow[n\to\infty]{\mathrm{a.s.}} 0,
\end{equation}
where \(\le\) is understood componentwise and \(\Phi\) denotes the cumulative distribution function of the standard Gaussian law on
\(\mathbb R^{m!-1}\).
\end{theorem}
The proof can be found in Appendix \ref{appendix:proofs}. 
\begin{remark}
    Note that \eqref{eq:eig} does not hold if we include the last $m!$-th component $1\{\Pi(X_{B(s)})=\pi_{m!}\}$ into $\mathbf{Y}_{s}\;.$ To see this, note that the minimum eigenvalue  of $\Var(\mathbf Y_{s}  )$ would then always be 0. In fact, for the vector of all ones $\mathbf{1}=(1,\ldots,1)^\top$ ,
    $$\Var\!\big(\mathbf 1^\top \mathbf Y_{s}  \big)
=\mathbf 1^\top \Var(\mathbf Y_{s}  )\mathbf 1.$$
However, $\mathbf 1^\top \mathbf Y_{s}=1$ a.s, and thus $\Var\!\big(\mathbf 1^\top \mathbf Y_{s}  \big)=0\;.$ Hence, $\Var(\mathbf Y_{s}  )$ is not strictly positive definite, but there is a vector which achieves equality. 
\end{remark}



\section{Covariance Estimation}\label{sec:cov}
From the previous theorem, we know that. 
\[
S_n=\sum_{v_s\in\mathcal V_n}\Big(\mathbf Y_{s}-\E[\mathbf Y_{s}  ]\Big),
\qquad
\Sigma_n=\Var(S_n  ).
\]
In practical applications, $\Sigma_n$ is unknown and it must be estimated. 
By denoting with \(\widetilde{\mathbf Y}_{s}:=\mathbf Y_{s}-\E[\mathbf Y_{s}  ]\), the covariance expands as 
\begin{align*}
    \Sigma_n
=&\Var\!\left(\sum_{v_s\in\mathcal V_n}\widetilde{\mathbf Y} _{s}\right)\\
=&\sum_{v_s\in\mathcal V_n}\sum_{v_t\in\mathcal V_n}
\Cov\!\big(\widetilde{\mathbf Y} _{s},\widetilde{\mathbf Y} _{t}  \big)\;.
\end{align*}
By regrouping the pairs \((v_s,v_t)\) according to their graph distance \(d_n(v_s,v_t)=h\) yields
\begin{align*}
    \Sigma_n
=&\sum_{h\ge 0}\ \sum_{v_s\in\mathcal V_n}\ \sum_{v_t\in\mathcal V_n^{\partial}(v_s;h)}
\Cov\!\big(\widetilde{\mathbf Y} _{s},\widetilde{\mathbf Y} _{t}  \big)\\
=&\sum_{h\ge 0} \Omega_n(h)\;,
\end{align*}
where we set 
\begin{align}
\label{eq:Omega}
     \Omega_n(h)=\sum_{v_s\in\mathcal V_n}\ \sum_{v_t\in\mathcal V_n^{\partial}(u;h)}
\Cov\!\big(\widetilde{\mathbf Y} _{s},\widetilde{\mathbf Y} _{t}  \big)\;.
\end{align}
Since the conditional means \(\E[\mathbf Y_{s}  ]\) are unknown, we replace
\(\widetilde{\mathbf Y} _{s}=\mathbf Y _{s}-\E[\mathbf Y _{s}  ]\) by its empirical counterpart
\(\mathbf Y _{s}-\bar{\mathbf Y}_n\), where \(\bar{\mathbf Y}_n:=n^{-1}\sum_{v_s\in\mathcal V_n}\mathbf Y _{s}\).
For each distance \(h\ge 0\), we then define the (unnormalized) empirical shell contribution
\begin{equation}\label{eq:Omega_hat_raw}
\widehat\Omega_n(h)
:=
\sum_{v_s\in\mathcal V_n}\ \sum_{v_t\in\mathcal V_n^{\partial}(v_s;h)}
\big(\mathbf Y _{s}-\bar{\mathbf Y}_n\big)\big(\mathbf Y _{t}-\bar{\mathbf Y}_n\big)^\top.
\end{equation}
This matrix is the sample analogue of \(\Omega_n(h)\) in \eqref{eq:Omega}. In particular, in view of the
decomposition \(\Sigma_n=\sum_{h\ge 0}\Omega_n(h)\), we estimate \(\Sigma_n\) by the kernel-truncated sum
\begin{equation}\label{eq:Sigma_hat_raw}
\widehat\Sigma_n
:=
\sum_{h\ge 0}\omega\!\left(\frac{h}{b_n}\right)\widehat\Omega_n(h),
\end{equation}
where \(b_n\to\infty\) is a bandwidth and \(\omega:[0,\infty)\to\mathbb R\) is a kernel with compact support
(e.g.\ the Bartlett kernel \(\omega(x)=(1-|x|)\mathbf 1_{[0,1]}(x)\)).

Finally, note that \(\widehat\Sigma_n\) estimates the covariance of the \emph{unnormalized} sum \(S_n\), and thus
typically diverges in norm as \(n\to\infty\). If one instead wishes to estimate the covariance matrix of the
normalized sum \(n^{-1/2}S_n\), which converges to a finite (non-degenerate) limit, one rescales by \(n^{-1}\) and
sets
\[
\widehat V_n:=\frac{1}{n}\widehat\Sigma_n,
\quad\text{and}\quad
\widehat V_n \approx \Var\!\left(n^{-1/2}S_n  \right)=\frac{1}{n}\Sigma_n.
\]
A sufficient condition to ensure consistency of $\hat{V}_n$  is to impose the following assumption (see Assumption 4.1 in \cite{KOJEVNIKOV2021882})

\begin{assumption} There exists $p>4$ such that 
    \begin{itemize}
    \item[(i)] 
    $\lim\limits_{n\to \infty} \sum_{h\geq 1}|\omega_n(h)-1|\delta_n^{\partial}(h)\alpha^{1-\frac{2}{p}}_n(h)=0 \quad\text{a.s.}$
    \item [(ii)] 
    $ \lim\limits_{n\to \infty} n^{-1} \sum_{h\geq 0}c_n(h, b_n,2) \alpha^{1-\frac{4}{p}}_n(h)=0\quad \text{a.s.}$
    \end{itemize}
\end{assumption}

\section{Test statistics}\label{sec:test_statistics}

By definition, the vector $\mathbf{p} = \big(p(\pi_1),\ldots,p(\pi_{m!})\big)^\top$
represents a point in the $m!$-dimensional unit simplex
\begin{equation}
\mathcal{S} =\left\{\mathbf{x} = (x_1,\ldots,x_{m!}) \in \mathbb{R}^{m!}:x_i \ge 0,\;\sum_{i=1}^{m!} x_i = 1\right\}.
\end{equation}

This simplex constitutes the sample space of \emph{compositional data}, that is, multivariate observations describing parts of a whole. 
Such data contain only relative information and are subject to the constraints of non-negativity and constant sum. 
As a consequence, classical statistical methods developed for unconstrained data in ordinary Euclidean space are generally not suitable for compositional data. 
Ignoring the particular geometry of the simplex may lead to misleading or spurious conclusions. For a detailed discussion of these issues, see 
the seminal work of \cite{aitchison1982statistical} and subsequent treatments of compositional data analysis such as \cite{pawlowsky2015modeling}, \cite{greenacre2018compositional}, and \cite{filzmoser2021advances}.

A principled way to address this issue is to analyze log-ratios between the components of the composition. 
Log-ratio transformations establish a one-to-one correspondence between the simplex and a real Euclidean space defined by a 
suitable coordinate representation of the original data. This transformation makes it possible to apply standard statistical techniques 
in the transformed space while preserving the relevant information contained in the composition. In addition, log-ratio methods 
possess desirable invariance properties, such as independence from the overall scale of the data (for example, whether the components 
sum to 1 or to 100) and from the number of parts forming the composition.

Accordingly, the test for spatial independence is performed not on the probabilities of the ordinal patterns themselves, but on 
their log-ratio representation. In particular, the additive log-ratio (ALR) coordinate vector consists of $m!-1$ coordinates obtained from the probabilities of the ordinal patterns $\mathbf{p}$ as follows:
\begin{equation}
\label{eq:alr}
\mathrm{ALR}(\mathbf{p}) =
\left(
\log\!\left(\frac{p(\pi_1)}{p(\pi_{m!})}\right),
\ldots,
\log\!\left(\frac{p(\pi_{m!-1})}{p(\pi_{m!})}\right)
\right)^\top.
\end{equation}

The ALR representation requires selecting one component of the composition as a reference and placing it in 
the denominator of the log-ratios. Without loss of generality, we take the frequency of the last symbol as the reference component. 
Choosing a different denominator simply corresponds to permuting the parts of the composition, which mathematically results in a linear 
transformation between the corresponding alr coordinate vectors.

Let us denote by  $\mathbf 1\in\mathbb R^{m!-1}$  the vector of ones. Under spatial independence, i.e., under $\mathcal H_0$, the ordinal-pattern
distribution is uniform, 
so that $\mathrm{ALR}(\frac{1}{m!}\mathbf 1)=\mathbf 0$.
We therefore consider the plug-in statistic
\begin{equation}\label{eq:alr_est}
\mathrm{ALR}(\hat{\mathbf p}_n)
=
\Bigg(
\log\!\Big(\tfrac{\hat p_n(\pi_1)}{\hat p_n(\pi_{m!})}\Big),\ldots,
\log\!\Big(\tfrac{\hat p_n(\pi_{m!-1})}{\hat p_n(\pi_{m!})}\Big)
\Bigg)^\top\;.
\end{equation}
Equivalently,
\begin{align*}
\mathrm{ALR}(\hat{\mathbf p}_n)
=
\left[\begin{array}{c}
     \log \hat p_n(\pi_1) \\
     \vdots
     \\
     \log \hat p_n(\pi_{m!-1})
\end{array}\right]
-
\log \hat p_n(\pi_{m!})\,\mathbf 1 .
\end{align*}

Next, we derive its asymptotic distribution under $\mathcal H_0$. Under the assumptions of Theorem~\ref{thm:cond_vector_clt}, we have under $\mathcal H_0$
\begin{equation}
\label{eq:clt_phat}
\sqrt n\left(
\left[\begin{array}{c}
     \widehat p_n(\pi_1)\\
     \vdots\\
     \widehat p_n(\pi_{m!-1})
\end{array}\right]
-\tfrac{1}{m!}\mathbf 1
\right)
\ \xrightarrow{\mathcal D}\
\mathcal N(\mathbf 0,\,V)\;,
\end{equation}
where  $V\in\mathbb R^{m!-1\times m!-1}$ is the asymptotic covariance matrix.
Moreover, we reinsert the last coordinate by
\[
\widehat p_{n}(\pi_{m!}):=1-\sum_{j=1}^{m!-1}\widehat p_n(\pi_j),
\quad
p(\pi_{m!}):=\tfrac{1}{m!}.
\]

Consider the map $\mathrm{ALR} :\mathbb R^{m!-1}\to\mathbb R^{m!-1}$ given by
\begin{equation}
\label{eq:alr_map}
\mathrm{ALR}(\mathbf r)
:=
\left[\begin{array}{c}
\log\!\big(r_1/g(\mathbf r)\big)\\
\vdots\\
\log\!\big(r_{m!-1}/g(\mathbf r)\big)
\end{array}\right],
\quad
g(\mathbf r):=1-\sum_{j=1}^{m!-1}r_j\;.
\end{equation}

The Jacobian matrix $D\left[\mathrm{ALR}\right](\mathbf r)\in\mathbb R^{m!-1\times m!-1}$ has entries
\begin{align*}
&\frac{\partial}{\partial r_j}\Big(\log r_i-\log \left(g(\mathbf r)\right)\Big)
=
\frac{\delta_{ij}}{r_i}+\frac{1}{g(\mathbf r)},\\
&i,j\in\{1,\ldots,m!-1\},
\end{align*}
where $\delta_{ij}$ denotes the Kronecker delta. Evaluating at $\tfrac{1}{m!}\mathbf 1$ yields
\[
D\left[ \mathrm{ALR}\right](\tfrac{1}{m!}\mathbf 1)_{ij}
=
m!(\delta_{ij}+1),\]
hence,
\[
J:=D\left[\mathrm{ALR}\right](\tfrac{1}{m!}\mathbf 1)
=
m!\Big(I_{m!-1}+\mathbf 1\mathbf 1^\top\Big)\;.
\]
Therefore, by the multivariate delta method applied to~\eqref{eq:clt_phat} and the map~\eqref{eq:alr_map},
\begin{equation}
\label{eq:clt_alr}
\sqrt n\Big(\mathrm{ALR}(\widehat{\mathbf p}_n)-\mathrm{ALR}(\tfrac{1}{m!}\mathbf 1)\Big)
\ \xrightarrow{\mathcal D}\
\mathcal N\!\Big(\mathbf 0,\ J V J^\top\Big),
\end{equation}
In particular, under $\mathcal H_0$ the ALR statistic is asymptotically centered and has covariance $J V J^\top$.

Let $\widehat V_n$ be a consistent estimator of $V$ under $\mathcal H_0$. Then, the Mahalanobis distance 
\begin{equation}
\label{eq:wald_alr}
L_n
:=
n\cdot
\mathrm{ALR}(\widehat{\bm p}_n)^\top
\Big(J\,\widehat V_n\,J^\top\Big)^{-1}
\mathrm{ALR}(\widehat{\bm p}_n)\;
\end{equation}
follows a chi-squared distribution with $m!-1$ degrees of freedom, i.e. under $\mathcal H_0$,
\begin{equation}
\label{eq:convergence_L_n}
    L_n \ \xrightarrow{\mathcal D}\ \chi^2_{m!-1}\;,
\end{equation}
see \cite{manly2024multivariate}.

We will reject $\mathcal H_0$ at a confidence level $\alpha$ whenever $L_n\ge q_{1-\alpha}$, where $q_{1-\alpha}$ denotes the $(1-\alpha)$-quantile of the $\chi^2_{m!-1}$ distribution. The estimator $\hat{V}_n$ is constructed via \eqref{eq:Sigma_hat_raw} by replacing $\bar{\mathbf{Y}}_u=\frac{1}{m!}\mathbf{1}$ for all $u\in \mathcal{V}_n\;,$ and normalizing by $1/n\;.$

\section{Monte Carlo Simulations}\label{sec:monte}
\paragraph*{Size of the test.}
We first assess the finite-sample behaviour of the test statistic $L_n$ under
$\mathcal{H}_0$.  By~\eqref{eq:convergence_L_n}, $L_n \xrightarrow{d} \chi^2_{m!-1}$
 and we investigate how quickly this asymptotic approximation
becomes accurate.  We generate i.i.d.\ spatial observations for sample sizes
$n \in \{500, 2000, 5000\}$ and compute $L_n$ over $R = 10 000$ Monte Carlo
replications.  Throughout, the ordinal-pattern dimension is fixed at $m = 3$,
so the limiting null distribution is $\chi^2_5$, which has theoretical mean~5,
variance~10, and median approximately~4.35.

Table~\ref{tab:qstat_summary} reports the empirical mean, variance, and median of $L_n$
for each sample size.  Across all values of $n$, the three summary statistics
are in close agreement with their theoretical counterparts under $\chi^2_5$
(mean~5, variance~10, median~4.35), and the agreement improves monotonically
with $n$.  Notably, even at the smallest sample size considered ($n = 500$),
the empirical mean (4.84) and median (4.31) already match the theoretical
values to within 3\%, demonstrating that the chi-squared centering is well
captured in moderate samples.  The empirical variance converges steadily from
8.32 to 9.65 as $n$ grows from 500 to 5000, reflecting the expected
improvement in the asymptotic approximation.
 
The QQ-plots and histograms in Figure~\ref{fig:under_null} provide a complementary
visual assessment.  At every sample size, the bulk of the empirical
distribution aligns closely with the $\chi^2_5$ law, and the agreement
improves rapidly with $n$: by $n = 2 000$ it is visually excellent across
the entire support, and at $n = 5 000$ the empirical distribution is
essentially indistinguishable from the asymptotic limit.  Taken together,
these results confirm that the $\chi^2_{m!-1}$ approximation is already
reliable at the moderate sample sizes that are typical in spatial
applications, and becomes sharper as $n$ increases.

\begin{table}[ht]
\centering
\large
\begin{tabular}{ccccc}
\hline
$n$ & $R$ & mean$(L_n)$ & var$(L_n)$ & median$(L_n)$ \\
\hline
500  & 10000 & 4.8412 & 8.3185 & 4.3118 \\
2000 & 10000 & 4.9156 & 9.2907 & 4.2753 \\
5000 & 10000 & 4.9316 & 9.6481 & 4.3143 \\
\hline
\end{tabular}
\caption{Empirical mean, variance, and median of $L_n$ under $\mathcal H_0$, based on $R=10 000$ Monte Carlo replications.}
\label{tab:qstat_summary}
\end{table}

\begin{center}
\begin{figure*}
    \centering
    \includegraphics[width=1\linewidth]{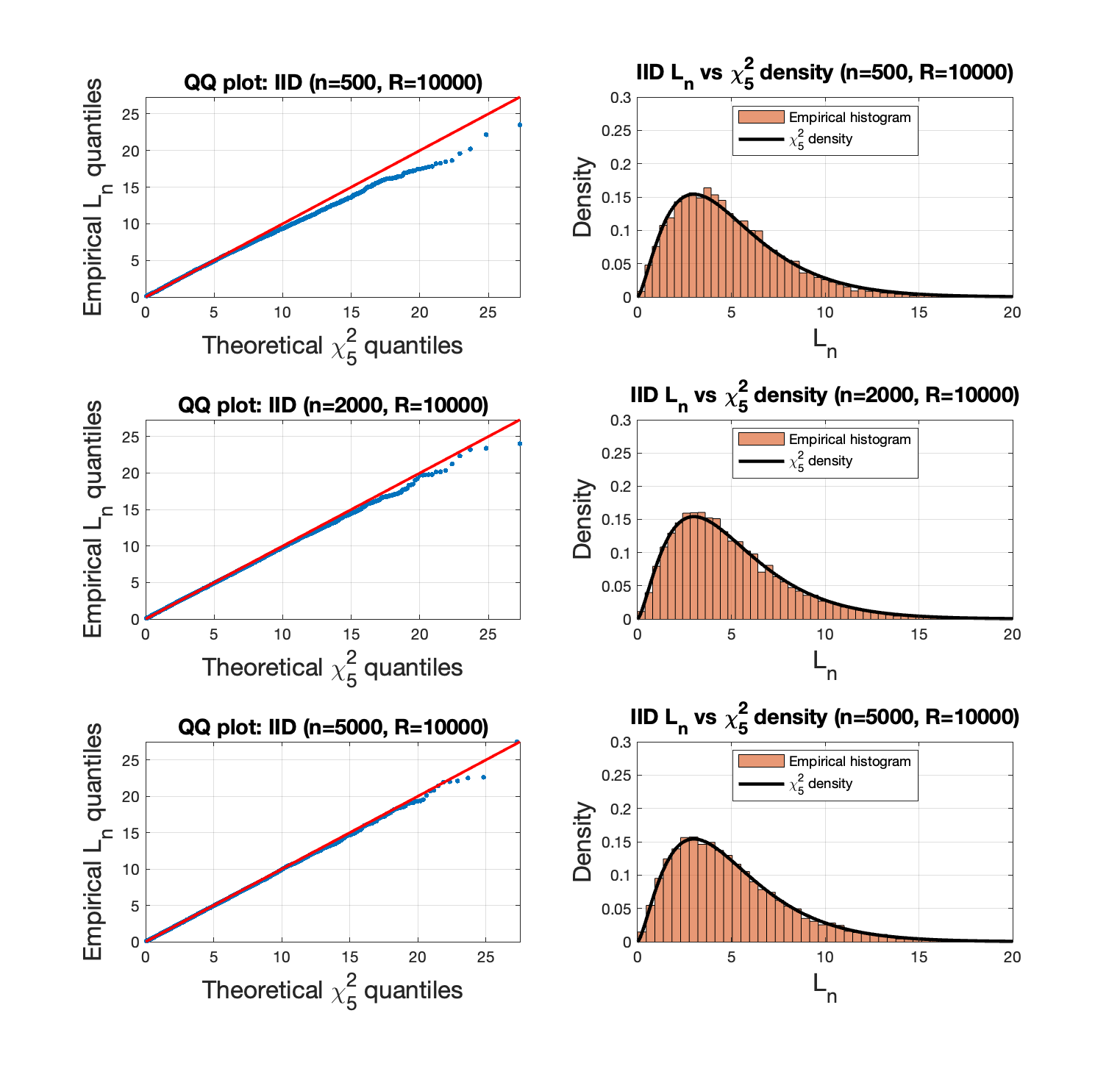}
    \caption{ QQ plot for different sample sizes (left). 
    Empirical density function versus theoretical density (right). Plots are obtained from $R=10000$ Monte Carlo replications. 
    }
    \label{fig:under_null}
\end{figure*}
\end{center}
\paragraph*{Power of the test. }
We assess the finite-sample power of the proposed test under the spatial
autoregressive (SAR) model of the previous example.  For each sample size
$n \in \{500, 2000, 5000\}$ we construct a nearest-neighbour
graph on the observed spatial locations, with number of neighbors $2$ or $3$.
The spatial weight matrix $W$ is the associated row-normalised adjacency
matrix.  For a dependence parameter $\rho \in \{0, 0.1, \ldots, 0.9\}$ we
generate observations from the reduced-form model
\[
  \mathbf{X} = (I_n - \rho W)^{-1}\boldsymbol{\varepsilon},
  \qquad
  \boldsymbol{\varepsilon} \sim \mathcal{N}(0, I_n).
\]
In order to examine the robustness of the test to departures from Gaussianity
and linearity, we also consider two nonlinear transformations of $\mathbf{X}$:
the sine-transformed model $\mathbf{X}' = \sin(\mathbf{X})$ and the
log-absolute-value model $\mathbf{X}' = \log|\mathbf{X}|$.  For each
configuration $(n, \rho)$ we run $R = 10 000$ replications, fix the ordinal
pattern length at $m \in \{3, 4\}$, compute the ALR-based statistic $L_n$, and
reject $\mathcal{H}_0$ at the 5\% nominal level using the $\chi^2_{m!-1}$
critical value.
 
Figure~\ref{fig:power_all_models} displays the empirical power curves as a function of
$\rho$ for all combinations of model, $m$, and $n$.  Several conclusions
emerge.
 
\paragraph*{Size control.}
At $\rho = 0$, which corresponds to spatial independence, all configurations
yield rejection rates close to the nominal 5\% level, confirming the good
finite-sample size calibration.
 
\paragraph*{Monotone power.}
The empirical rejection rate increases monotonically with $\rho$ in every
scenario, as expected for a consistent test: stronger spatial dependence
produces larger deviations of the ordinal-pattern distribution from
uniformity, which the statistic successfully detects.
 
\paragraph*{Effect of sample size.}
The influence of $n$ is substantial, particularly at intermediate values of
$\rho$.  For the linear SAR model with $m = 3$, the test achieves power
close to 1 for $\rho \geq 0.5$ when $n = 5 000$, while at $n = 500$ a
comparable power level requires $\rho \approx 0.8$.  This gap narrows as $n$
increases, illustrating the consistency of the test.
 
\paragraph*{Robustness to nonlinear transformations.}
A key feature of ordinal-pattern methods is their invariance under strictly
monotone transformations: if $f$ is monotone increasing, the ordinal pattern
of $(X_{u_0}, \ldots, X_{u_{m-1}})$ coincides with that of
$(f(X_{u_0}), \ldots, f(X_{u_{m-1}}))$.  Since $\sin$ and $\log|\cdot|$ are
not globally monotone, spatial dependence is distorted in a nonlinear fashion,
which a priori could reduce power.  Nevertheless, Figure~\ref{fig:power_all_models} shows
that the power curves under both transformed models are qualitatively similar
to those of the linear SAR, with only modest reductions in power at
intermediate $\rho$.  This finding illustrates the robustness of the test to
nonlinear data-generating mechanisms, a practically important property that
linear spatial-dependence tests do not in general possess.
 
\paragraph*{Effect of ordinal-pattern dimension.}
Comparing $m = 3$ and $m = 4$, we observe that increasing $m$ tends to reduce
power for smaller $n$.  This is a consequence of the larger
number of degrees of freedom in the limiting distribution: under $m = 4$, the
$\chi^2_{m!-1} = \chi^2_{23}$ critical value is substantially higher than the
$\chi^2_5$ value used for $m = 3$, so a larger signal is required to exceed
it.  At large $n$ (e.g., $n = 5 000$) the power for $m = 4$ is larger than for $m = 3$, suggesting that the additional
information captured by longer patterns eventually compensates for the increased degrees of freedom.  In practice, we recommend $m = 3$ for moderate
sample sizes and $m = 4$ when $n$ is large.

\begin{figure*}[!t]
\centering

\begin{subfigure}[t]{0.49\textwidth}
    \centering
    \includegraphics[width=\linewidth]{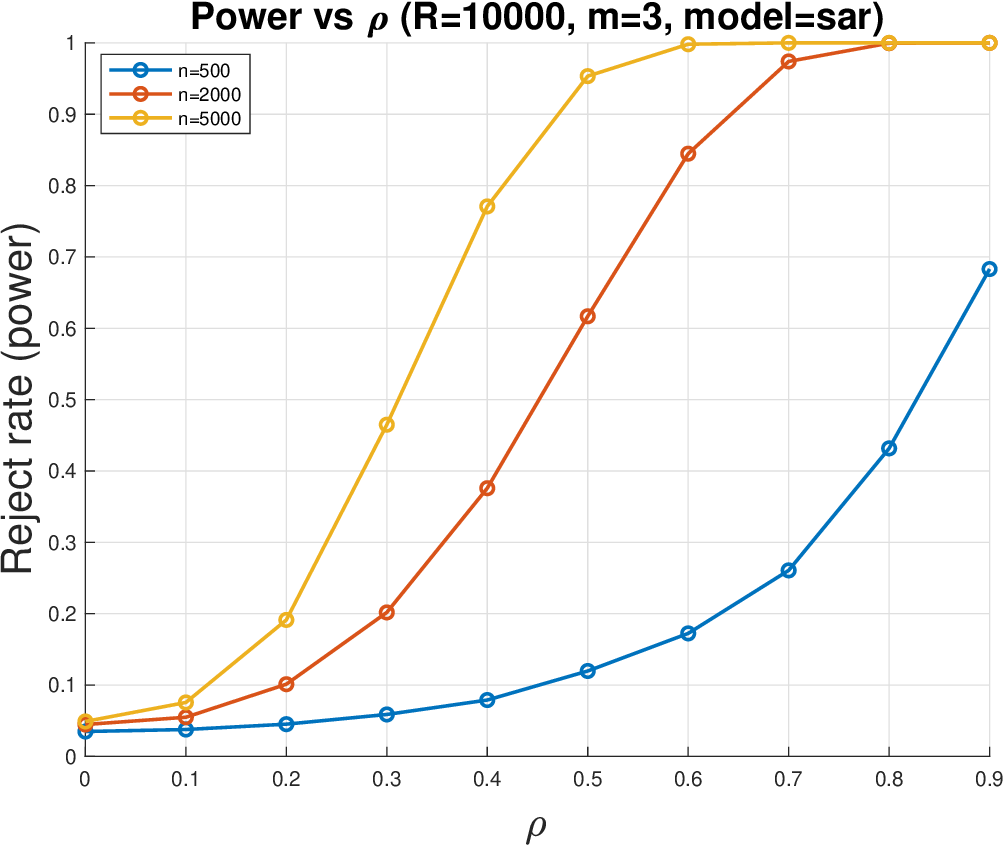}
    \caption{Linear SAR, $m=3$.}
    \label{fig:power_linear_m3}
\end{subfigure}\hfill
\begin{subfigure}[t]{0.49\textwidth}
    \centering
    \includegraphics[width=\linewidth]{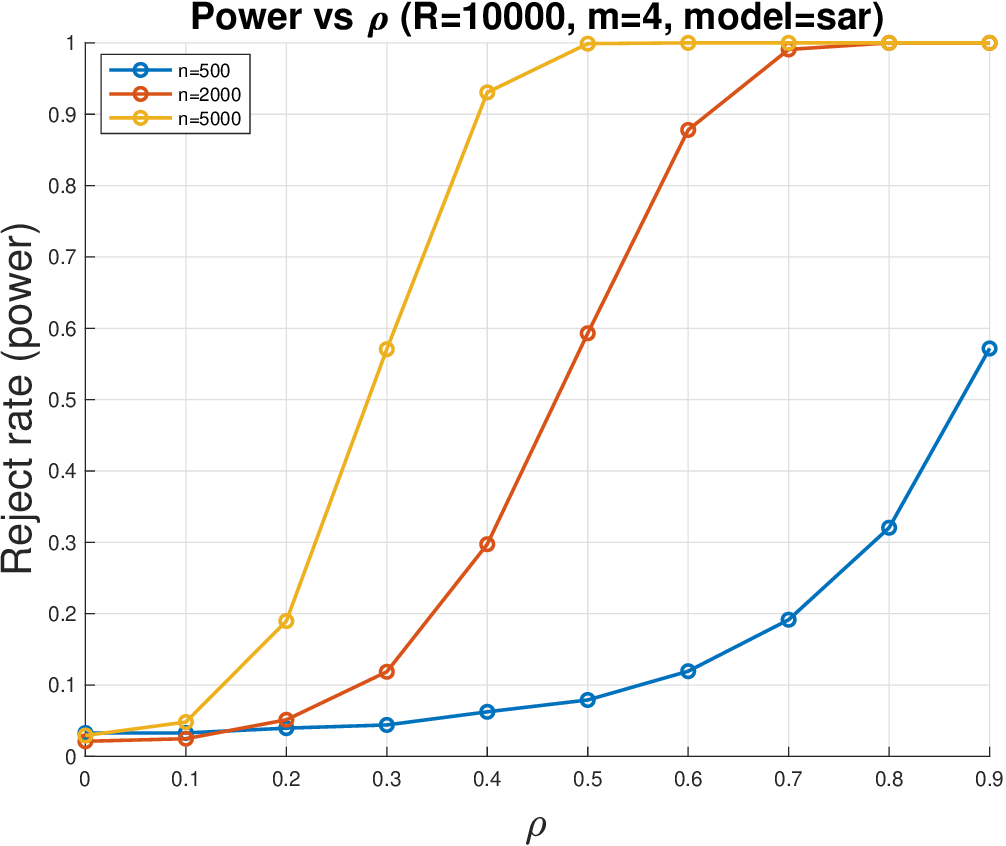}
    \caption{Linear SAR, $m=4$.}
    \label{fig:power_linear_m4}
\end{subfigure}

\vspace{0.35cm}

\begin{subfigure}[t]{0.49\textwidth}
    \centering
    \includegraphics[width=\linewidth]{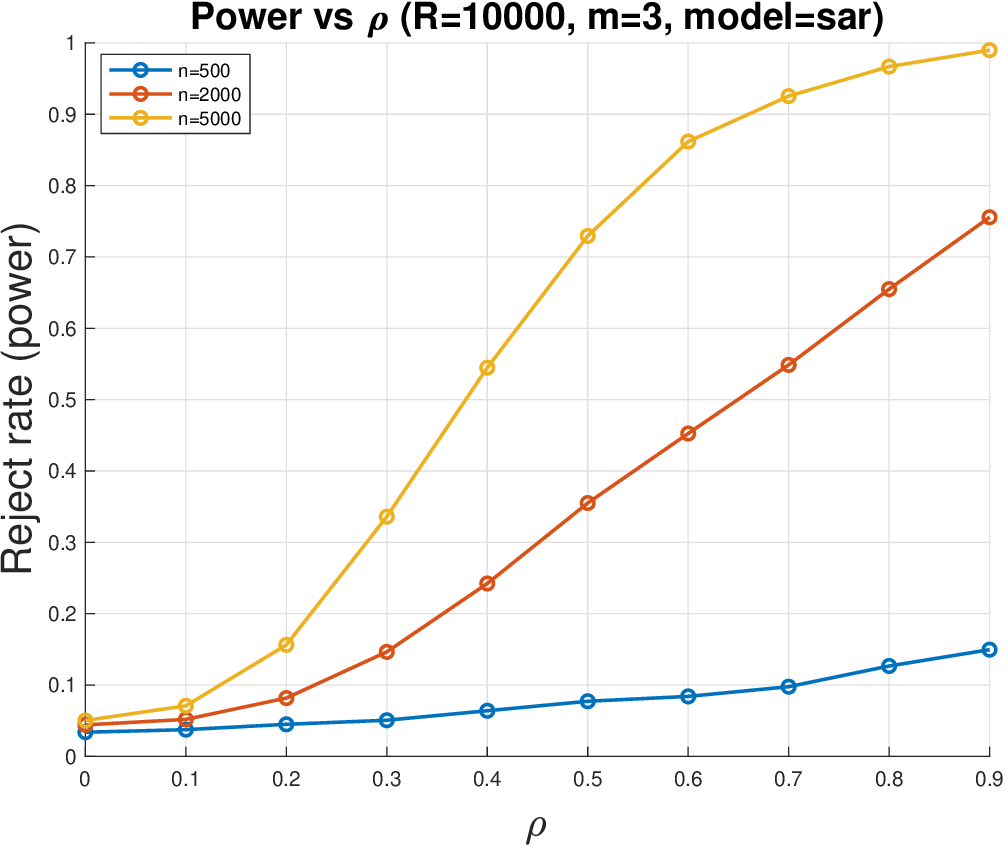}
    \caption{$\sin$-transformed SAR, $m=3$.}
    \label{fig:power_sin_m3}
\end{subfigure}\hfill
\begin{subfigure}[t]{0.49\textwidth}
    \centering
    \includegraphics[width=\linewidth]{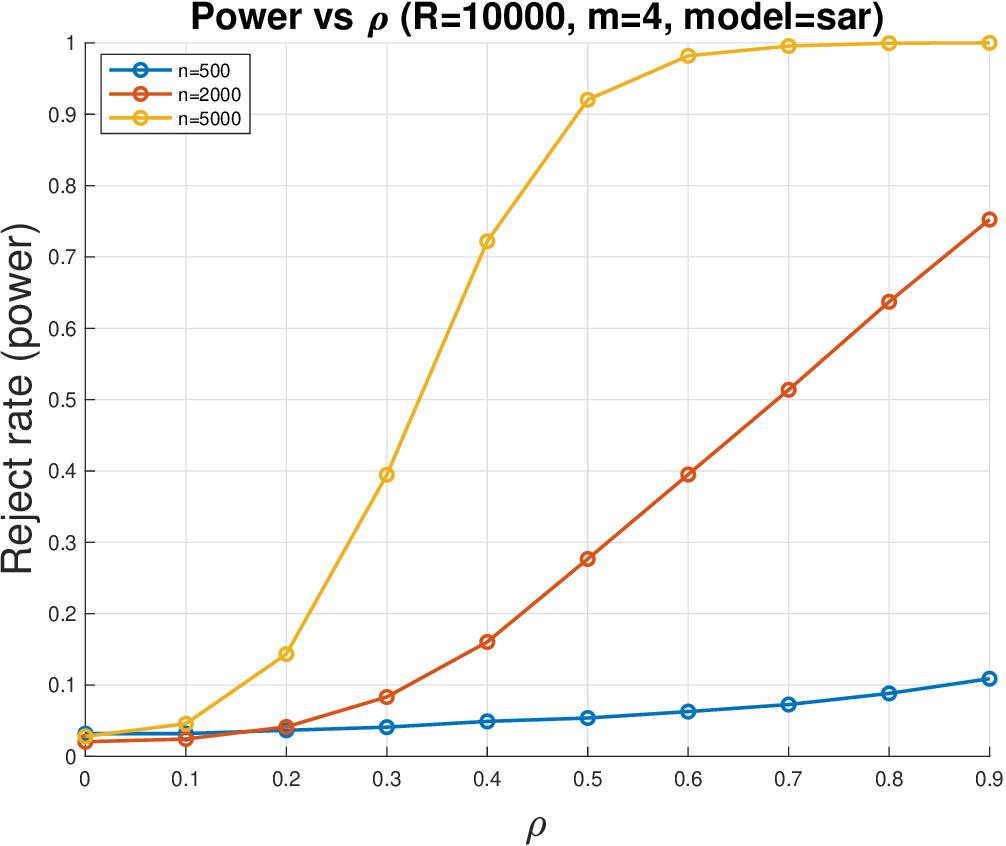}
    \caption{$\sin$-transformed SAR, $m=4$.}
    \label{fig:power_sin_m4}
\end{subfigure}

\vspace{0.35cm}

\begin{subfigure}[t]{0.49\textwidth}
    \centering
    \includegraphics[width=\linewidth]{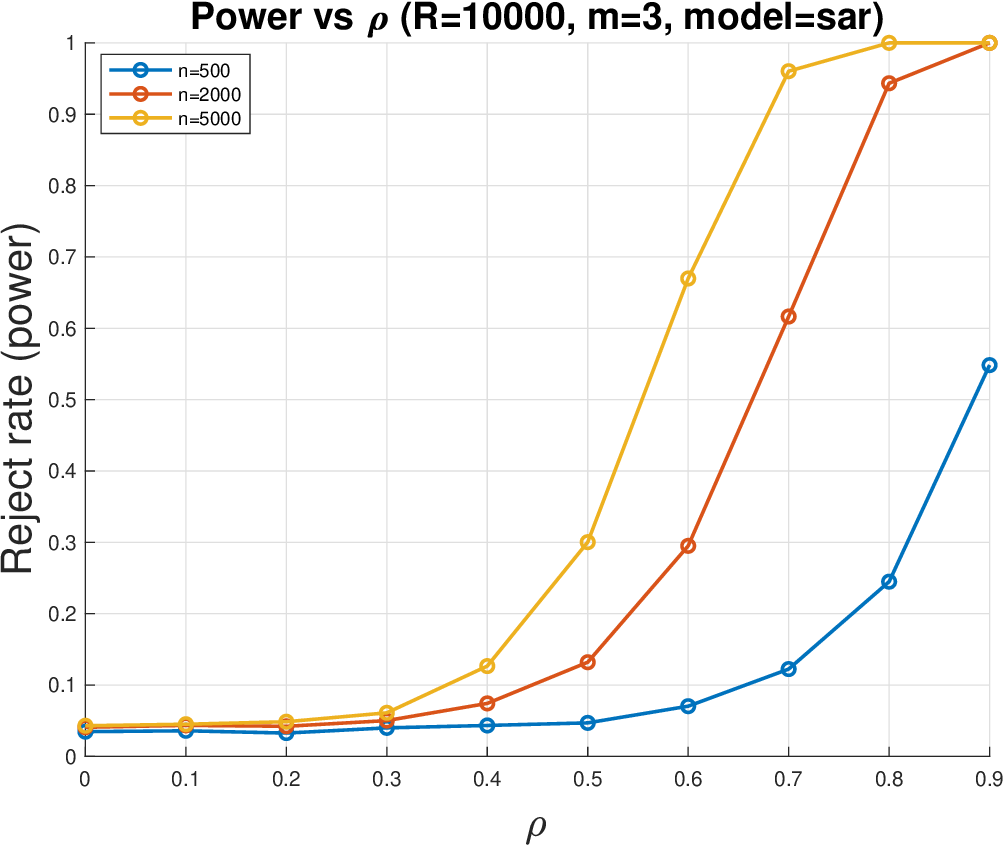}
    \caption{$\log|\cdot|$-transformed SAR, $m=3$.}
    \label{fig:power_logabs_m3}
\end{subfigure}\hfill
\begin{subfigure}[t]{0.49\textwidth}
    \centering
    \includegraphics[width=\linewidth]{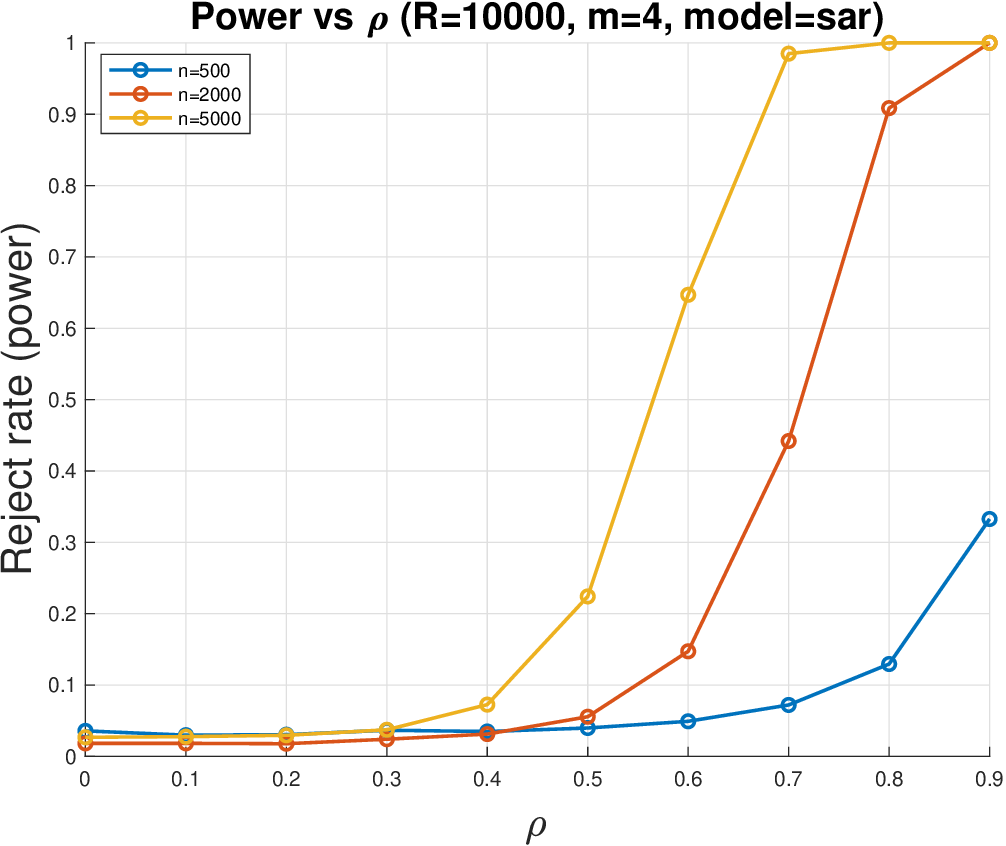}
    \caption{$\log|\cdot|$-transformed SAR, $m=4$.}
    \label{fig:power_logabs_m4}
\end{subfigure}

\caption{Empirical power curves based on $R=10 000$ Monte Carlo repetitions for ordinal pattern dimensions $m\in\{3,4\}$ under three data-generating mechanisms: the linear SAR model $X=(I_n-\rho W)^{-1}\varepsilon$, the nonlinear model $X=\sin((I_n-\rho W)^{-1}\varepsilon)$, and the nonlinear model $X=\log |(I_n-\rho W)^{-1}\varepsilon|$. }
\label{fig:power_all_models}
\end{figure*}

\section{Conclusions}
\label{sec:conclusions}

This paper has introduced the first ordinal-pattern-based test of spatial
independence for data observed on irregular, non-lattice point clouds
$\mathcal{V}_{n}\subset\mathbb{R}^{2}$.
The central methodological step is to replace the rectangular $2\times2$
lattice windows used by SOP-based procedures with local nearest-neighbour
configurations: for each location $v_s\in\mathcal{V}_{n}$, the ordinal
pattern $\Pi(X_{B(s)})\in\mathcal{S}_{m}$ encodes only the relative
order of the $m$ observations in the neighborhood block $B(s)$, retaining
no information about absolute values or inter-location distances.

Under the null hypothesis of spatial independence, the ordinal patterns are
i.i.d.\ and uniformly distributed over the symmetric group $\mathcal{S}_{m}$,
a distributional characterization that holds for any continuous marginal $F$
and is entirely free of parametric assumptions.
We exploit this fact to construct the Wald test statistic $L_{n}$ based on
the additive log-ratio transformation of the empirical ordinal-pattern
frequencies.
Under the mixing and graph-sparsity conditions of Theorem~\ref{thm:cond_vector_clt},
$L_{n}$ converges in distribution to a $\chi^{2}_{m!-1}$ random variable,
providing a simple, asymptotically pivotal procedure with no tuning parameters
beyond the embedding dimension~$m$.

The finite-sample properties of the test were assessed through an extensive
Monte Carlo study with $R=10 000$ replications.
Under the null, the empirical distribution of $L_{n}$ converges rapidly to the
$\chi^{2}_{m!-1}$ limit: for $m=3$ (so that the reference distribution is
$\chi^{2}_{5}$, with theoretical mean $5$, variance $10$, and median $\approx
4.35$), the empirical moments are already close to their theoretical values at
$n=500$ and indistinguishable from them at $n=5000$.
Under the alternative, the test achieves correct size control at $\rho=0$ and
exhibits monotone power that increases with both the dependence parameter
$\rho$ and the sample size $n$.
Crucially, the ordinal-pattern statistic detects dependence not only in the
linear Gaussian SAR model $X=(I_n-\rho W)^{-1}\varepsilon$, but also in the
nonlinear transformations $X=\sin\bigl((I_n-\rho W)^{-1}\varepsilon\bigr)$ and
$X=\log\bigl|(I_n-\rho W)^{-1}\varepsilon\bigr|$, for both $m=3$ and $m=4$.
This robustness to monotone transformations is an intrinsic feature of
ordinal-pattern methods and constitutes a clear advantage over classical
tests based on linear spatial autocorrelation.

Several directions for future research are suggested by this work.
First, the asymptotic theory developed here relies on the graph-based
$\alpha$-mixing framework of~\citet{KOJEVNIKOV2021882}; extending the
validity of $L_n$ to broader classes of spatial dependence, including
long-range or non-mixing processes, would widen the scope of the method.
Second, the choice of the embedding dimension $m$ involves a bias--variance
trade-off: larger $m$ provides finer discrimination between ordinal patterns
and potentially higher power, but at the cost of a larger number of
degrees of freedom and a slower convergence of $\hat{p}_n$ to its limit.
Data-driven selection of $m$, analogous to lag-order selection in time
series, is an open problem worthy of investigation.
Third, it would be natural to consider marked or multivariate extensions of
the test, in which each location carries a vector of observations rather than
a scalar, or to incorporate the spatial configuration of the locations
themselves as additional information.
Finally, applications to real spatial datasets -- such as environmental
monitoring networks, epidemiological surveillance data, or urban economic
data -- would demonstrate the practical value of the proposed procedure
in settings where irregular spatial supports are the rule rather than
the exception.

\section*{Data Availability Statement}

The data that support the findings of this study are available from the corresponding author upon reasonable request.

\section*{Acknowledgments }
Manuel Ruiz Mar\'{\i}n is grateful to the support of a grant from the Spanish Ministry of Science, Innovation and Universities PID2022-136252NB-I00 founded by MICIU/AEI/10.13039/501100011033 and by the European Regional Development Fund (FEDER, EU). \\
Mariano Matilla‑García is grateful to the support of a grant from the Spanish Ministry of Science, Innovation and Universities PID2022-136547NB-I00, founded by MICIU/AEI/10.13039/501100011033 and by European Regional Development Fund (FEDER, EU).

\vspace{1cm}
The following article has been submitted to \emph{Chaos: An Interdisciplinary Journal of Nonlinear Science}. After it is published, it will be found at \url{https://publishing.aip.org/publications/journals/special-topics/chaos/classical-ordinal-patterns-and-beyond/}
\nocite{*}
\bibliographystyle{plainnat}
\bibliography{aipsamp} 
\newpage
\appendix

\section{Proof of Theorem 1}

\label{appendix:proofs}

\begin{lemma}\label{lem:alpha_mixing_OP_local}
Let $\mathcal V_n$ be a set of spatial locations, and let 
$d_n(\cdot,\cdot)$ be the corresponding graph-distance. 
Let $\{X_{s}\}_{v_s\in\mathcal V_n}$ be a sequence of random variables on a common probability space and define the corresponding ordinal-pattern indicator
\[
Y_{s}:=\mathbf 1\!\left\{\Pi\big(X_{B(s)}\big)=\pi\right\},
\qquad v_s\in\mathcal V_n,
\]
for some fixed $\pi\in \mathcal S_m$. 
Denote by $\alpha^X(h)$ and $\alpha^Y(h)$ the spatial mixing coefficients of $X$ and $Y$ respectively.

Then the following inequality holds
\[
\alpha^Y_{n}(h)\ \le\ \alpha^X_{n}\big(\max\{h-2r_n,\,0\}\big),
\qquad h\ge 0.
\]
where $r_n=\max_{v_s\in \mathcal{V}_n}\max_{v_u\in B(s)}d_n(v_s,v_u)$.

In particular, if $\alpha^X_{n}(h)\to 0$ as $h\to\infty$ 
and $\sup_n r_n<\infty$, then $\alpha^Y_{n}(h)\to 0$ as $h\to\infty$, i.e.\ $(Y_{s})_{v_s\in \mathcal{V}_n}$ is strongly mixing. 
\end{lemma}

\begin{proof}
Fix $n\ge 1$ and subsets $A,B\subset\mathcal V_n$.
Set
\[
U_n(A):=\bigcup_{v_s\in A} B(s),
\qquad
U_n(B):=\bigcup_{v_t\in B} B(t).
\]
Since each $Y_{s}$ is a measurable function of $(X_{s})_{v_s\in B(s)}$, we have
\[
\sigma(Y_{A})\subset \sigma(X_{U_n(A)}),
\qquad
\sigma(Y_{B})\subset \sigma(X_{U_n(B)}).
\]
and therefore,

\begin{equation}
\label{eq:alpha_monotone}
\sup_{\substack{G\in \sigma(Y_A)\\ H\in \sigma(Y_B)}}
\big|\Cov(\mathbf 1_G,\mathbf 1_H)\big|
\;\le\;
\sup_{\substack{G\in \sigma(X_{U_n(A)})\\ H\in \sigma(X_{U_n(B)})}}
\big|\Cov(\mathbf 1_G,\mathbf 1_H)\big|.
\end{equation}
Assume now that $d_n(A,B)\ge h$. Take arbitrary $u\in U_n(A)$ and $v\in U_n(B)$
Then there exist $a\in A$ and $b\in B$ such that $u\in B(a)$ and $v\in B(b)$.
By the triangle inequality and the definition of $r_n$,
\begin{align*}
    d_n(u,v)
\ge& d_n(a,b)-d_n(a,u)-d_n(b,v)\\
\ge& h-r_n-r_n
= h-2r_n.
\end{align*}

Hence $d_n(U_n(A),U_n(B))\ge h-2r_n$.
Therefore, by the definition of $\alpha^X_n(\cdot)$,
\[
\sup_{\substack{G\in \sigma(X_{U_n(A)})\\ H\in \sigma(X_{U_n(B)})}}
\big|\Cov(\mathbf 1_G,\mathbf 1_H)\big|
\le
\alpha^X_{n}\big(\max\{h-2r_n,0\}\big).
\]
Combining this with \eqref{eq:alpha_monotone} yields
\[
\sup_{\substack{G\in \sigma(Y_A)\\ H\in \sigma(Y_B)}}
\big|\Cov(\mathbf 1_G,\mathbf 1_H)\big|
\le
\alpha^X_{n}\big(\max\{h-2r_n,0\}\big).
\]
Taking the supremum over all $A,B\subset\mathcal V_n$ such that $d_n(A,B)\ge h$  we obtain that $\alpha^Y_n(h)\leq\alpha_n^X(h)$.
The final implication follows immediately when $\sup_n r_n<\infty$ and $\alpha^X_n(h)\to 0$ as $h\to\infty$.
\end{proof}

\begin{proof}[Proof of Theorem \ref{thm:cond_vector_clt}]
We argue via the Cram\'er--Wold device. Fix $\mathbf a\in\mathbb{R}^{m!-1}$ with $\|\mathbf a\|_2=1$ and define the scalar projection
\[
S_n(\mathbf a):=\mathbf a^\top S_n
=\sum_{v_s\in\mathcal V_n}\Big(\mathbf a^\top\mathbf Y_{s}-\E[\mathbf a^\top\mathbf Y_{s}  ]\Big).
\]
It suffices to show that $S_n(\mathbf a)$ satisfies a CLT for every such $\mathbf a$; the vector statement then
follows from the Cram\'er--Wold device.

To this end we apply  Theorem~3.2 of \cite{KOJEVNIKOV2021882} to the array $\big(\mathbf a^\top\mathbf Y_{s}\big)_{v_s\in\mathcal V_n}$. For this theorem to hold it requires Assumptions~2.1, 3.3, and 3.4 in  \cite{KOJEVNIKOV2021882}  to be satisfied. The rest of the proof is then dedicated to verify these assumptions. Assumption~2.1 is a network/density condition on $G$ and is guaranteed here by
\eqref{eq:A_mixing_small}; indeed, their Proposition~2.2 shows that Assumption~2.1 holds
whenever the underlying array is strongly mixing with coefficients controlled as in \eqref{eq:A_mixing_small}.
Moreover, Lemma~\ref{lem:alpha_mixing_OP_local} ensures that strong mixing of $(X_{s})_{v_s \in \mathcal{V}_s}$ transfers to
$(\mathbf Y_{s})_{v_s\in \mathcal{V}_n}$ (with the same mixing coefficients), since $\mathbf Y_{s}$ is a measurable functional of
$X_{s}$. Consequently, Assumption~2.1 follows from \eqref{eq:A_mixing_small}.

Assumption~3.3 is immediate because $\mathbf Y_{s}\in\{0,1\}^{m!-1}$, hence
$\big|\mathbf a^\top\mathbf Y_{s}\big|\le \|\mathbf a\|_1\le \sqrt{m!-1}$ for all $s$.
The main task is Assumption 3.4, for which it suffices to prove that there exists
$c>0$ such that
\[
\Var\!\big(S_n(\mathbf a)  \big)\ \ge\ c\,n
\qquad\text{a.s.}
\]
In fact, the above inequality alongside our Assumption \ref{assumption1} imply that Condition ND of \cite{KOJEVNIKOV2021882} hold (see their Assumption 3.4).  
Using variance decomposition and $x\ge -|x|$, we obtain
\begin{align*}
&\Var\!\big(S_n(\mathbf a)  \big)
\\&=
\sum_{v_s\in\mathcal V_n}\Var\!\big(\mathbf a^\top \mathbf Y_{s}  \big)
\\&+\sum_{\substack{v_s,v_t\in\mathcal V_n\\ s\neq t}}
\Cov\!\big(\mathbf a^\top \mathbf Y_{s},\mathbf a^\top \mathbf Y_{t}  \big)\notag\\
&\ge
\sum_{v_s\in\mathcal V_n}\Var\!\big(\mathbf a^\top \mathbf Y_{s}  \big)
\\&-\sum_{\substack{v_s,v_t\in\mathcal V_n\\ s\neq t}}
 \Big|\Cov\!\big(\mathbf a^\top \mathbf Y_{s},\mathbf a^\top \mathbf Y_{t}  \big)\Big| .
\label{eq:var_lower_bound_split_better}
\end{align*}
We treat the diagonal and off-diagonal terms separately. 

For the diagonal term, note that
\[
\Var\!\big(\mathbf a^\top \mathbf Y_{s}  \big)
=\mathbf a^\top \Var(\mathbf Y_{s}  )\mathbf a.
\]
Hence, a sufficient condition for a uniform lower bound is condition \eqref{eq:eig}. 
Under \eqref{eq:eig} and $\|\mathbf a\|_2=1$,
\[
\Var\!\big(\mathbf a^\top \mathbf Y_{s}  \big)\ge C_1
\Longrightarrow
\sum_{v_s\in\mathcal V_n}\Var\!\big(\mathbf a^\top \mathbf Y_{s}  \big)\ \ge\ n C_1.
\]

For the off-diagonal term, we use the fact that $\{\mathbf Y_{s}\}_{v_s\in\mathcal V_n}$ is conditionally strongly mixing
given $ $: there exists a sequence $\{\alpha_n(h)\}_{h\ge 0}$ such that
\[
\sup_{\substack{G\in \sigma(Y_A)\\ H\in \sigma(Y_B)}}
\big|\Cov(\mathbf 1_G,\mathbf 1_H)\big| \le\ \alpha_n(h)
\quad\text{when}\quad d_n(A,B)\ge h,
\]
where $d_n(\cdot,\cdot)$ denotes graph distance. This property will be ensured by Lemma~\ref{lem:alpha_mixing_OP_local}.
 Assuming strongly mixing means that we have control over the covariance as the distance among nodes increase. In particular by Proposition 2.2 of \cite{KOJEVNIKOV2021882}
\begin{align*}
    & \Cov\!\big(\mathbf a^\top \mathbf Y_{s},\mathbf a^\top \mathbf Y_{t}  \big)  \\&\leq \| \mathbf a^\top \|_\infty \|\mathbf a^\top \|_\infty \alpha_n (d_n(v_s,v_t))\;.
\end{align*}
Since $\|\mathbf{a}\|_2=1$,
\begin{align*}
   & \sum_{\substack{v_s,v_t\in\mathcal V_n\\ s\neq t}}
\Big| \Cov \!\big(\mathbf a^\top \mathbf Y_{s},\mathbf a^\top \mathbf Y_{t,n}  \big)\Big|  \leq \sum_{\substack{v_s,v_t\in\mathcal V_n\\ s\neq t}}\alpha_n (d_n(s,t))\;.
\end{align*}
We go by diagonals. For $r\geq 1$, define
$$U_n(h):=\#\{ (v_s,v_t)\in \mathcal{V}_n\times \mathcal{V}_n\,:\, d_n(v_s,v_t)=h\}\;.$$
Then, 
\begin{align*}
   & \sum_{\substack{v_s,v_t\in\mathcal V_n\\ s\neq t}}
 \Cov\!\big(\mathbf a^\top \mathbf Y_{s},\mathbf a^\top \mathbf Y_{t}  \big) \leq  \sum_{h\geq 1}\alpha_n (h)U_n(h)\;.
\end{align*}
Hence, via assumption \eqref{eq:A_mixing_small}
\[
\Var\!\big(S_n(\mathbf a)\big)\ge
\big(C_2-C_1\big)n
=:cn
\qquad\text{a.s.}
\]
with \(c>0\). This proves the required uniform lower bound on the variance, and therefore verifies Assumption~3.4, which finishes the proof of the theorem.
\end{proof}

\section{Auxiliary Results}
\label{appendix:auxiliary}
\begin{definition}
Let \((\Omega,\mathcal F,\mathbb P)\) be a probability space, and let
\(\mathcal A,\mathcal B\subset \mathcal F\) be two sub-\(\sigma\)-fields.
Their maximal correlation coefficient is defined by
\[
\rho(\mathcal A,\mathcal B)
:=
\sup_{\substack{U\in L^2(\mathcal A),\,V\in L^2(\mathcal B)\\
\Var(U),\,\Var(V)>0}}
\frac{|\Cov(U,V)|}{\sqrt{\Var(U)\Var(V)}}.
\]
\end{definition}

\begin{definition}
Let \(H_1,H_2\subset L^2(\Omega,\mathcal F,\mathbb P)\) be two closed linear subspaces.
Their maximal correlation coefficient is defined by
\[
\rho(\mathcal A,\mathcal B)
:=
\sup_{\substack{U\in H_1,\,V\in H_2\\
\Var(U),\,\Var(V)>0}}
\frac{|\Cov(U,V)|}{\sqrt{\Var(U)\Var(V)}}.
\]

\end{definition}

The following result corresponds to Theorem 1 in \cite{KolRoz60}. 
\begin{theorem}
\label{thm:KR_gaussian_maxcorr}
Let \(\{\xi_i\}_{i\in I}\) and \(\{\eta_j\}_{j\in J}\) be two jointly Gaussian families of square-integrable random variables in \(L^2(\Omega,\mathcal F,\mathbb P)\).
Let
\[
\mathcal F_\xi:=\sigma(\xi_i:i\in I),
\qquad
\mathcal F_\eta:=\sigma(\eta_j:j\in J),
\]
and let
\[
H_\xi:=\overline{\mathrm{span}\{\xi_i:i\in I\}}^{\,L^2},
\qquad
H_\eta:=\overline{\mathrm{span}\{\eta_j:j\in J\}}^{\,L^2}.
\]
Then
\[
\rho(\mathcal F_\xi,\mathcal F_\eta)=\rho(H_\xi,H_\eta).
\]
\end{theorem}
In plain words, for jointly Gaussian families, the maximal correlation between the
\(\sigma\)-fields they generate is equal to the maximal correlation between
their closed linear spans in \(L^2\). 
\begin{remark}
\label{remark:Gaussian}
Theorem~\ref{thm:KR_gaussian_maxcorr} is particularly useful when
\[
\xi=(\xi_1,\dots,\xi_p),\qquad \eta=(\eta_1,\dots,\eta_q)
\]
are finite-dimensional jointly Gaussian vectors. In that case,
\[
\rho(\mathcal F_\xi,\mathcal F_\eta)
=
\sup_{a\in\mathbb R^p,\; b\in\mathbb R^q}
\mathrm{Corr}(a^\top \xi,\; b^\top \eta).
\]
\end{remark}

\begin{proposition}\label{prop:gaussian_alpha_bound}
Let \(p,q\in\mathbb N\), and let
\[
\xi=(\xi_1,\dots,\xi_p),\qquad \eta=(\eta_1,\dots,\eta_q)
\]
be jointly Gaussian random vectors. For \(A\subset\{1,\dots,p\}\) and \(B\subset\{1,\dots,q\}\), write
\[
\xi_A:=(\xi_u)_{u\in A},\qquad \eta_B:=(\eta_v)_{v\in B},
\]
and let
\[
\mathcal F_A:=\sigma(\xi_u:u\in A),\qquad
\mathcal F_B:=\sigma(\eta_v:v\in B)
\]
be two $\sigma$-fields.
Assume that there exist constants \(M>0\) and \(\lambda_0>0\) such that
\[
|\Cov(\xi_u,\eta_v)|\le M
\qquad\text{for all }u=1,\dots,p,\ v=1,\dots,q,
\]
and the minimum eigenvalues satisfy
\[
\lambda_{\min}\!\bigl(\Cov(\xi)\bigr)\ge \lambda_0,
\qquad
\lambda_{\min}\!\bigl(\Cov(\eta)\bigr)\ge \lambda_0.
\]
Then there exists a constant \(C=C(p,q,\lambda_0)>0\) such that, for all
\(A\subset\{1,\dots,p\}\) and \(B\subset\{1,\dots,q\}\), and \[
\alpha(\mathcal A,\mathcal B)
:=
\sup_{A\in\mathcal A,\;B\in\mathcal B}
\big|\mathbb P(A\cap B)-\mathbb P(A)\mathbb P(B)\big|,
\] it holds
\[
\alpha(\mathcal F_A,\mathcal F_B)\le C\,M.
\]
More precisely, for 
\[
\alpha(\mathcal F_A,\mathcal F_B)
\le
\frac{1}{4}\frac{\sqrt{|A||B|}}{\lambda_0}\,M.
\]
\end{proposition}
\begin{proof}
Fix \(A\subset\{1,\dots,p\}\) and \(B\subset\{1,\dots,q\}\), and set
\[
\Sigma_{AB}:=\Cov(\xi_A,\eta_B).
\]
By Remark~\ref{remark:Gaussian},
\[
\rho(\mathcal F_A,\mathcal F_B)
=
\sup_{a\in\mathbb R^{|A|},\,b\in\mathbb R^{|B|}}
\mathrm{Corr}(a^\top \xi_A,\; b^\top \eta_B).
\]
Hence
\[
\rho(\mathcal F_A,\mathcal F_B)
=
\sup_{a,b}
\frac{|a^\top \Sigma_{AB} b|}
{\sqrt{\Var(a^\top \xi_A)\Var(b^\top \eta_B)}}.
\]
Since every entry of \(\Sigma_{AB}\) is bounded by \(M\), we have
\[
\|\Sigma_{AB}\|_{2}
\le \|\Sigma_{AB}\|_{\mathrm F}
\le \sqrt{|A||B|}\,M.
\]
Therefore,
\[
|a^\top \Sigma_{AB} b|
\le \|\Sigma_{AB}\|_{2}\,\|a\|_2\,\|b\|_2
\le \sqrt{|A||B|}\,M\,\|a\|_2\,\|b\|_2.
\]

Now, for each fixed pair \((A,B)\), since \(\xi_A\) and \(\eta_B\) are, by assumption, finite non degenerate-dimensional Gaussian vectors and the smallest eigenvalue is bounded away from zero by a constant independent of $A$ and $B$, there exists a constant \(c_{A,B}>0\) such that
\[
\Var(a^\top \xi_A)\ge c\|a\|_2^2,
\qquad
\Var(b^\top \eta_B)\ge c\|b\|_2^2
\]
for all \(a\in\mathbb R^{|A|}\), \(b\in\mathbb R^{|B|}\). Hence
\[
\rho(\mathcal F_A,\mathcal F_B)\le C_{A,B}\,M
\]
for some constant \(C_{A,B}>0\). Since there are only finitely many choices of \(A\) and \(B\), we may take
\[
C:=\max_{A,B} C_{A,B}<\infty,
\]
and obtain
\[
\rho(\mathcal F_A,\mathcal F_B)\le C\,M
\qquad\text{for all }A,B.
\]
Finally, using the inequality
\[
\alpha(\mathcal F_A,\mathcal F_B)\le \frac14 \rho(\mathcal F_A,\mathcal F_B),
\]
we conclude that
\[
\alpha(\mathcal F_A,\mathcal F_B)\le \frac{C}{4}\,M.
\]
Absorbing the factor \(1/4\) into the constant completes the proof.
\end{proof}

\end{document}